\documentclass[a4paper, 11pt]{article}

\usepackage{mathrsfs}
\usepackage{epsfig}
\usepackage{amsmath}
\usepackage{amssymb}
\usepackage{latexsym}
\usepackage{amsfonts}
\usepackage{amsthm}
\usepackage[numbers,sort&compress]{natbib}
\usepackage{graphicx}
\ifx\pdfoutput\undefined  \else
 \fi
\usepackage[centerlast]{caption2}
\usepackage{color}

\newtheorem{definition}{Definition}

\newtheorem{lemma}{Lemma}
\newtheorem{theorem}{Theorem}
\newtheorem{observation}{Observation}

\numberwithin{figure}{section} \numberwithin{definition}{section}
\numberwithin{table}{section} \numberwithin{definition}{section}
\numberwithin{observation}{section} \numberwithin{lemma}{section}
\numberwithin{theorem}{section}

\begin{document}

\title{
{The crossing number of pancake graph $P_4$ is six} \footnote{The
research is supported by NSFC (60973014, 61170303)}
\author{
\ Yuansheng Yang\footnote {corresponding
author's email : yangys@dlut.edu.cn} , \ Bo Lv, \ Baigong Zheng, \ Xirong Xu, \ Ke Zhang \\
School of Computer Science and Technology\\
Dalian University of Technology, Dalian, 116024, P.R. China\\
%
}}

\date{}
\maketitle
\begin{abstract}

The {\it crossing number} of a graph $G$ is the least number of
pairwise crossings of edges among all the drawings of $G$ in the
plane. The pancake graph is an important topology for
interconnecting processors in parallel computers. In this paper, we
prove the exact value of the crossing number of pancake graph $P_4$
is six.

\end{abstract}

\section{Introduction}

The notion of {\it crossing number} is a central one for Topological
Graph Theory with long history, which means the minimum possible
number of edge crossings in a drawing of graph $G$ in the plane. In
recent years, because of its applications in various fields such as
discrete and computational geometry, VLSI theory, wiring layout
problems, and in several other areas of theoretical computer
science, the crossing number problem has been studied extensively by
mathematicians including Erd\H{o}s, Guy, Tur\'{a}n and Tutte,et
al.(see \cite{EG73,Guy60,Turan77,Tutte70}) However, the
investigation on the crossing number of graphs is an extremely
difficult problem. In 1973, Erd\H{o}s and Guy wrote, ``Almost all
questions that one can ask about crossing numbers remain unsolved.''
Actually, Garey and Johnson \cite{Garey83} proved that computing the
crossing number is NP-complete. Also, it's not surprising that the
exact crossing numbers are known only for a few families of
graphs(see \cite{Bruce95,Bruce07}). In most cases, to give the upper
and lower bounds is a more practical way(see
\cite{YangYS03,YangYS02,LinXH09}). As to a nice drawing of a graph
with the number of crossings that can hardly be decreased, it is
very difficult to prove that the number of crossings in this drawing
is indeed the crossing number of the graph we studied.

The pancake graph was proposed by Akers and Krishnameurthy in
\cite{Akers89} as a special case of Cayley graphs. It not only
possesses several attractive features just like hypercubes, such as
symmetry properties and high fault tolerant, but also offers three
significant advantages over hypercubes: a lower degree, a smaller
diameter and average diameter. Therefore, there are more and more
research about pancake graphs recently. In \cite{LHH05}, Lin, Huang
and Hsu proved that the n-dimensional pancake graph $P_n$ is super
connected if and only if $n\neq 3$. In addition,  Deng and Zhang
proved that the automorphism group of the pancake graph $P_n$ is the
left regular representation of the symmetric group $S_n$ for
$n\geq5$ in \cite{DZ12}. More research about pancake graph can be
found in \cite{C11,LTHHH09,NB09,S05,DZ12}.

In \cite{SV94}, S\'{y}kora and Vrt'o proved approximative values of
crossing number of the pancake graphs. However, their results are
valuable only when the dimension $n$ is large enough. Yet there is
little study of the exact crossing number of pancake graphs when $n$
is small, which is of theoretical importance and practical value. In
this paper, we prove that the crossing number of pancake graph $P_4$
is exactly six.

\section{Notations and basic lemmas}

\indent \indent Let $G$ be a simple connected graph with vertex set
$V(G)$ and edge set $E(G)$. For $S\subseteq E(G)$, let $[S]$ be the
subgraph of $G$ induced by $S$. Let $P_{v_1v_2\cdots v_n}$ be the
\textit{path} traversing from $v_1$ to $v_n$ with $n$ vertices. Let
$C_{v_1v_2\cdots v_nv_1}$ be the \textit{circle} with $n$ vertices
from $v_1$ to $v_n$.

A drawing of $G$ is said to be a {\it good} drawing, provided that
no edge crosses itself, no adjacent edges cross each other, no two
edges cross more than once, and no three edges cross in a point. It
is well known that the crossing number of a graph is attained only
in {\it good} drawings of the graph. So, we always assume that all
drawings throughout this paper are good drawings.

For a drawing $D$ of a graph $G$, let $\nu(D)$ be the number of
crossings in $D$. In a drawing $D$, if an edge is not crossed by any
other edge, we say that it is {\it clean} in $D$.

For two disjoint subsets of an edge set $E$, say $A$ and $B$, the
number of the crossings formed by an edge in $A$ and another edge in
$B$ is denoted by $\nu_D(A,B)$ in a drawing $D$. The number of the
crossings that involve a pair of edges in $A$ is denoted by
$\nu_D(A)$. Then $\nu_D(A\cup B)=\nu_D(A)+\nu_D(B)+\nu_D(A,B)$ and
$\nu(D)=\nu_D(E)$.

\begin{definition}$($Pancake Graphs$)$
The n-dimensional pancake graph, denoted by $P_n$ and proposed by
Akers and Krishnameurthy, is a graph consisting of $n!$ vertices
labelled with $n!$ permutations on a set of the symbols $1, 2,
\cdots , n$. There is an edge from vertex $i$ to vertex $j$ if and
only if $j$ is a permutation of $i$ such that $i = i_{1}i_{2}\cdots
i_{k}i_{k+1}\cdots i_n$ and $j = i_k\cdots i_{2}i_{1}i_{k+1}\cdots
i_n$, where $2\leq k \leq n$.
\end{definition}
The pancake graphs $P_2, P_3$ and $P_4$ are shown in Figure
\ref{Figure P_n} for illustration.
\begin{figure}[ht]
\centering
\includegraphics[scale=1.0]{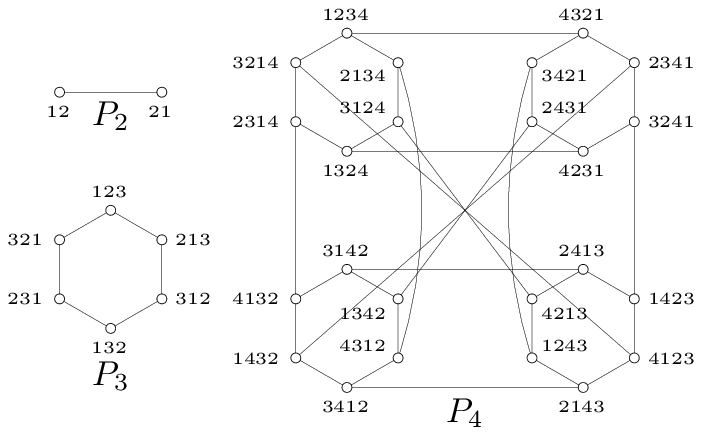}
\includegraphics[scale=1.0]{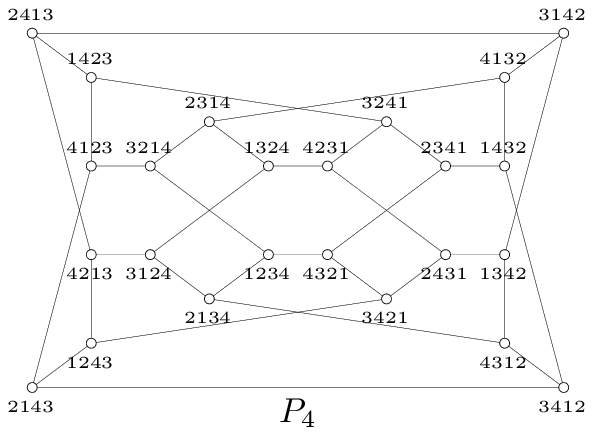}
\caption{\small{some drawings of $P_n$}} \label{Figure P_n}
\end{figure}

There are four 6-cycles $C_i (1\leq i\leq 4)$ in $P_4$. For $1\leq
i\leq 4$, the subgraph of $P_4$ induced by $V(P_4)-V(C_i)$ is
homeomorphic to graph $G_{12}$ shown in Figure \ref{Figure D_{120}}.

\begin{figure}[ht]
\centering
\includegraphics[scale=1.0]{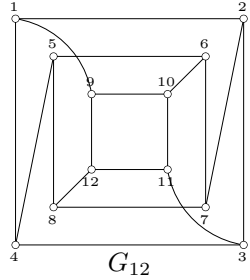}
\caption{\small{A drawing of $G_{12}$}} \label{Figure D_{120}}
\end{figure}

\begin{lemma}\label{Lemma D_{12}}
  Let $D$ be an arbitrarily drawing
of $G_{12}$, then $\nu(D)\geq 2.$
\end{lemma}

\begin{proof}
 Let $m$ be the smallest number of the edges of $G_{12}$
whose deletion from $G_{12}$ results in a planar subgraph
$G^{*}_{12}$ of $G_{12}$. $G^{*}_{12}$ has $12$ vertices, $18-m$
edges. Let $D^{*}_{12}$ be a planar drawing of $G^{*}_{12}$ and $p$
denote the number of faces in $D^{*}_{12}$. Then, according to the
Euler Polyhedron Formula,
\begin {eqnarray*}
12-(18-m)+p&=&2,\\
           p&=&8-m.
\end {eqnarray*}
Since all cycles in $G_{12}$ have length at least six except for
three 4-cycles, there are at most three 4-cycles in $G^{*}_{12}$ but
no 3-cycles. By counting the number of edges of each face in
$D^{*}_{12}$, we have
\begin {eqnarray*}
            3\times4+(8-m-3)\times 6&\leq &|E(G^{*}_{12})|= 2\times (18-m),\\
                           4m&\ge &6.
\end {eqnarray*}
It follows $m\ge$ 2. Hence $\nu(D)\geq 2$.\hspace{7.40cm}
\end{proof}

For $i=1,2,3$, let
$C^{4}_{i}=C_{v_{4i-3}v_{4i-2}v_{4i-1}v_{4i}v_{4i-3}}$.

\begin{lemma}\label{Lemma D'_{12}}
Let $D$ be a drawing of $G_{12}$, where at least one pair of
$4$-cycles crosses each other, then $\nu(D)\geq 3$.
\end{lemma}

\begin{proof}By contradiction. Suppose  $\nu(D)\leq 2$. Then
there is exact one pair of 4-cycles, say $C^4_1, C^4_2$, crosses
each other. Since $\nu(D)\leq 2$, no 4-cycle crosses itself, and
edges $v_1v_9,v_3v_{11}, v_6v_{10}$ and $v_8v_{12}$ are all clean.
Since vertices $v_9,v_{10},v_{11}$ and $v_{12}$ lie in outside of
$C^4_1$ and $C^4_2$, vertices $v_1,v_3,v_6$ and $v_8$ have to lie on
the bounder of same area. It follows at least one edge of
$v_1v_9,v_3v_{11}, v_6v_{10}$ and $v_8v_{12}$ is crossed,
$\nu(D)\geq 3$, a contradiction (See Figure\ref{Figure D_{12}}(1)).
\end{proof}

\begin{figure}[h]
\center
\includegraphics[scale=1.0]{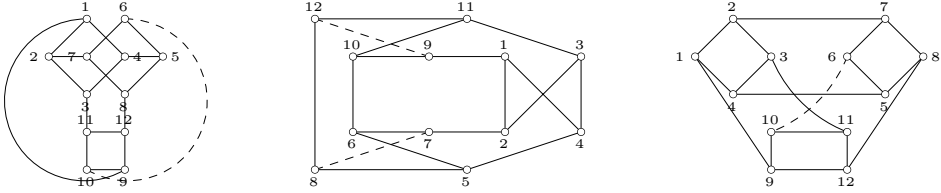}
\caption{Some drawings of $G_{12}$}\label{Figure D_{12}}
\end{figure}

\begin{lemma}\label{Lemma D''_{12}}
Let $D$ be a good drawing of $G_{12}$, where any pair of $4$-cycles
does not cross each other and any two $4$-cycles lie in the same
side of the third $4$-cycle, then $\nu(D)\geq 3$.
\end{lemma}

\begin{proof}By contradiction. Suppose  $\nu(D)\leq 2$.

\noindent  \textbf{Case 1.} There is at least one 4-cycle, say
$C^{4}_{1}$, crosses itself. Without loss of generality, we may
assume that $v_1v_4$ crosses $v_2v_3$. We show this situation in
Figure\ref{Figure D_{12}}(2). Since $\nu(D)\leq 2$, at least one
cycle of the disjoint cycles $C_{v_1v_9v_{10}v_6v_7v_2v_1}$ and
$C_{v_3v_{11}v_{12}v_{8}v_5v_4v_3}$, say cycle
$C_{v_1v_9v_{10}v_6v_7v_2v_1}$, does not cross itself. And at least
one cycle of cycles $C_{v_4v_3v_{11}v_{10}v_9v_1v_4}$ and
$C_{v_3v_{4}v_{5}v_{6}v_7v_2v_3}$, say cycle
$C_{v_4v_3v_{11}v_{10}v_9v_1v_4}$, does not cross itself, for they
only have one common edge. Considering the possible locations for
vertex $v_{12}$, we find at least one edge of
$\{v_{12}v_{11},v_{12}v_{9},v_{12}v_8,v_8v_7\}$ is crossed, since
edges $v_{11}v_{12}$, $v_9v_{12}$ and path $P_{v_{12}v_8v_7}$ can
not be in the same region. Hence, cycle
$C_{v_3v_{4}v_{5}v_{6}v_7v_2v_3}$ can not cross itself. Since
$\nu(D)\leq 2$, path $P_{v_5v_8v_{12}v_{11}}$ can be crossed at most
one time, edge $v_8v_{12}$ has to lie in outside of cycle
$C_{v_5v_4v_3v_{11}v_{10}v_6v_5}$. It follows edges $v_{12}v_9$ and
$v_8v_7$ are both crossed, $\nu(D)\geq 3$, a contradiction.

\noindent  \textbf{Case 2.} No 4-cycle crosses itself. Since
$\nu(D)\leq 2$, at least one of all the three pairs of 4-cycles, say
$C^4_{1}$ and $C^4_{2}$, satisfies the following conditions: the
edges between that pair of 4-cycles do not cross each other, and
they do not cross the pair of 4-cycles either. By symmetry, we may
assume $v_3,v_6$ lie in inside of cycle $C_{v_1v_2v_7v_8v_5v_4v_1}$
(See Figure\ref{Figure D_{12}}(3)). Since any pair of 4-cycles does
not cross each other and any two 4-cycles lie in the same side of
the third 4-cycle, 4-cycle $C^4_{3}$ has to lie in outside of cycle
$C_{v_2v_7v_6v_5v_4v_3v_2}$ or in inside of cycle
$C_{v_2v_7v_6v_5v_4v_3v_2}$. By symmetry, we may assume 4-cycle
$C^4_{3}$ lies in outside of cycle $C_{v_2v_7v_6v_5v_4v_3v_2}$. Then
edges $v_3v_{11}$ and $v_6v_{10}$ are crossed. Since $\nu(D)\leq 2$,
edges $v_1v_9$ and $v_8v_{12}$ are clean. It follows at least one
edges of $v_3v_{11}$ and $v_6v_{10}$ is crossed at least two times,
$\nu(D)\geq 3$, a contradiction.
\end{proof}

\section{Crossing number of $P_4$}

 \indent In Figure 3.1, we show a drawing of $P_4$ with 6
crossings. Hence, we have:

\begin{lemma} \label{lemma upper bound of P_4}
$cr(P_4)\leq 6$.
\end{lemma}
\begin{figure}[ht]
\centering
\includegraphics[scale=1.0]{P42.eps}
\includegraphics[scale=1.0]{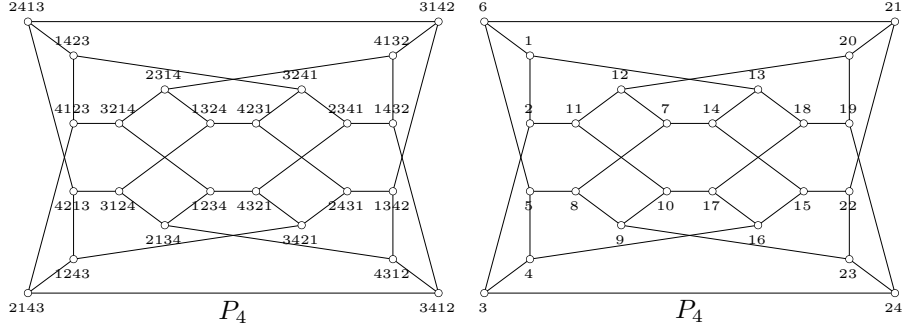}
\caption{\small{A good drawing of $P_4$ with 6 crossings}}
\end{figure}

In the rest of this section, we shall prove that the value of
$cr(P_4)$ is exactly equal to 6. We rename the vertices of $P_4$ as
shown in Figure 3.1.

For $i=1,2,3,4$, let

$$\begin{array}{rlll}
C^{i}_{p_4}&=C_{v_{6i-5}v_{6i-4}v_{6i-3}v_{6i-2}v_{6i-1}v_{6i}v_{6i-5}},\\
V^{i}_{p_4}&=V(C^{i}_{p_4}),\\
E^{i}_{p_4}&=E(C^{i}_{p_4}),\\
E^{i,j}_{p_4}&=\{uv:u\in V^{i}_{p_4}\wedge v\in V^{j}_{p_4} \},\\
E'^{i}_{p_4}&=E^{i}_{p_4}\cup\bigcup_{(1\leq j\leq 4)\wedge j\neq i} E^{i,j}_{p_4},\\
\overline{E'^{i}_{p_4}}&=E(P_4)-E'^{i}_{p_4}.\\
\end{array}$$

For convenience, we abbreviate
$$\begin{array}{rlll}
C_{i}&=C^{i}_{p_4},   V_{i}&=V^{i}_{p_4},
E_{i}&=E^{i}_{p_4},\\
E_{i,j}&=E^{i,j}_{p_4},   E'_{i}&=E'^{i}_{p_4},
\overline{E'_{i}}&=\overline{E'^{i}_{p_4}}.\\
\end{array}$$
Then we have the following important observation.
\begin{observation}
For $1\leq i\neq j\leq4$, let $u_av_a,u_bv_b\in E_{i,j}$,
$u_a,u_b\in V_i$ and the path between $u_a$ and $u_b$ is
$P_{u_au_cu_du_b}$ $($or $P_{u_bu_cu_du_a})$ on $C_i$. Then $u_c$
and $u_d$ are connected to different $6$-cycles except for $C_i$ and
$C_j$ $($See Figure 3.1$)$.
\end{observation}

\begin{figure}[ht]
\center
\includegraphics[scale=1.0]{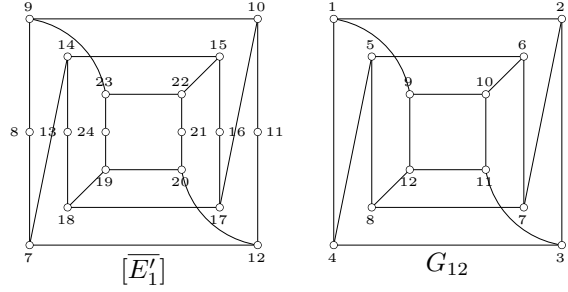}
\caption{$[\overline{E'_{i}}]$ is homeomorphic to $G_{12}$}
\label{Figure: G_{12}'}
\end{figure}

Since $[\overline{E'_{i}}]$ is homeomorphic to $G_{12}$ (See Figure
\ref{Figure: G_{12}'}). By Lemmas \ref{Lemma D_{12}}, \ref{Lemma
D'_{12}} and \ref{Lemma D''_{12}}, we have

\begin{lemma}\label{Lemma [E'_{i}]}
For $i=1,2,3,4,$\\
$1)$ Let $D$ be an arbitrarily drawing of
$[\overline{E'_{i}}]$, then $\nu(D)\geq 2.$\\
$2)$ Let $D$ be a drawing of $[\overline{E'_{i}}]$, where
at least one pair of $6$-cycles crosses each other, then $\nu(D)\geq 3.$\\
$3)$ Let $D$ be a drawing of $[\overline{E'_{i}}]$, where any pair
of $6$-cycles does not cross each other and any two 6-cycles lie in
the same side of the third $6$-cycle, then $\nu(D)\geq 3.$
\end{lemma}

\begin{lemma}\label{Lemma P_{2}}
 Let $D$ be a drawing of $P_4$, where at least two pairs
of $6$-cycles cross each other, then $\nu(D)\geq 6$.
\end{lemma}

\begin{proof}By contradiction. Suppose  $\nu(D)\leq 5$. Since
 each pair of 6-cycles crossing each other will produce at least two
crossings, there are at most two pairs of 6-cycles crossing each
other. By symmetry, there are two cases:

\noindent \textbf{Case 1.} $C_1$ crosses $C_2$ and $C_3$. There are
two subcases depending on $C_4$'s position.

\noindent \textbf{Case 1.1.} $C_4$ lies in outside of $C_2$ and
$C_3$ (See Figure\ref{Figure P_4-2}(1)). By Lemma\ref{Lemma
[E'_{i}]}, $\nu_D(\overline{E'_{1}})\geq 3$. Then, it follows
$\nu(D)\geq 3+4=7$, a contradiction.

\noindent \textbf{Case 1.2.} $C_4$ lies in inside of $C_3$($C_2$)
(See Figure\ref{Figure P_4-2}(2) and (3)). By Lemma\ref{Lemma
[E'_{i}]}, $\nu_D(\overline{E'_{1}})\geq 2$. Then, it follows
$\nu(D)\geq 2+4=6$, a contradiction.

\noindent \textbf{Case 2.} $C_1$ crosses $C_2$, $C_3$ crosses $C_4$.
By symmetry, we need only consider the case that $C_3$ and $C_4$ lie
in outside of $C_1$ and $C_2$ (See Figure\ref{Figure P_4-2}(4)). By
Lemma\ref{Lemma [E'_{i}]}, $\nu_D(\overline{E'_{1}})\geq 3$. Since
$\nu(D)\leq 5$, $C_1$ does not cross itself, any edge of
$\bigcup_{j=2,3,4} E_{1,j}$ is clean. It follows edges $v_{1}v_{13}$
and $v_4v_{16}$ are clean. Now at least one edge of $E_{1,4}$ is
crossed, it contradicts any edge of $\bigcup_{j=2,3,4} E_{1,j}$ is
clean (See Figure\ref{Figure P_4-2}(5)).

\begin{figure}[ht]
\center
\includegraphics[scale=1.0]{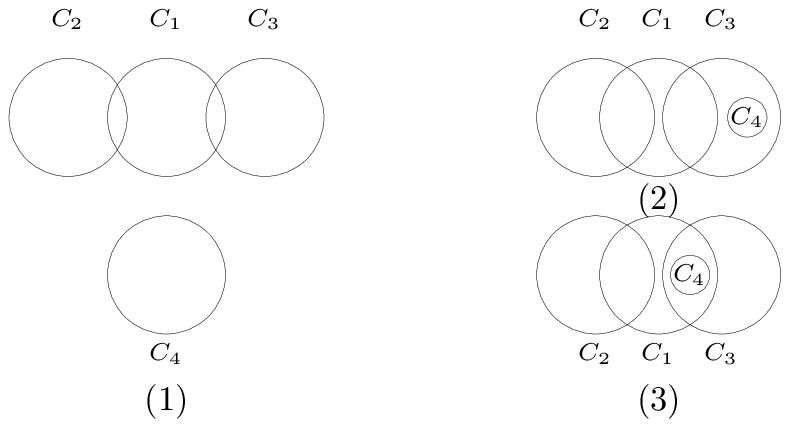}
\includegraphics[scale=1.0]{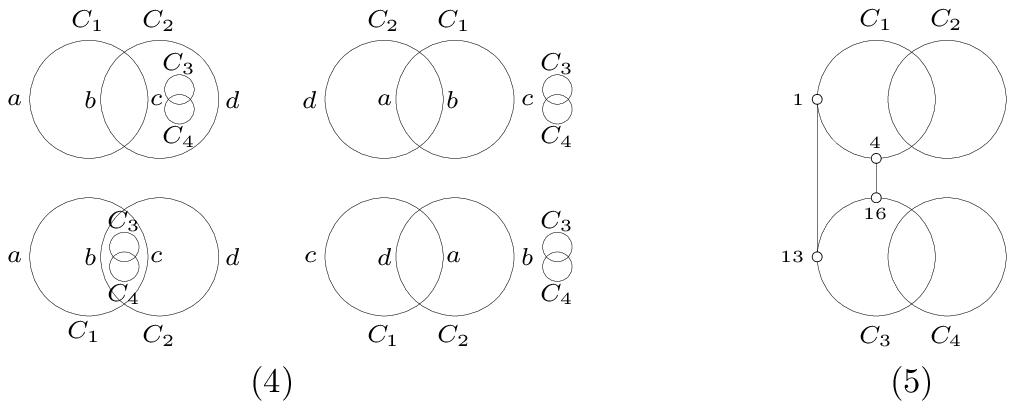}
\caption{Some Drawings of $P_4$, where just two pairs of 6-cycles
cross each other}\label{Figure P_4-2}
\end{figure}

\end{proof}

\begin{lemma}\label{Lemma P_{1}}
 Let $D$ be a drawing of $P_4$, where just one pair
of $6$-cycles crosses each other, then $\nu(D)\geq 6$.
\end{lemma}

\begin{proof}By contradiction. Suppose  $\nu(D)\leq 5$. By symmetry,
we need only consider the case that $C_3$ lie in outside of $C_1$
and $C_2$ (See Figure\ref{Figure P_4-1}(1)). There are three cases
depending on $C_4$'s position:

\noindent \textbf{Case 1.} $C_4$ lies in inside of $C_3$ (See
Figure\ref{Figure P_4-1}(2)). Then each edge of $E_{4,1}$ crosses
the edges of $E_{3}$ at least one time, each edge of $E_{4,2}$
crosses the edges of $E_{3}$ at least one time. It follows
$\nu(D)\geq 2+2+2=6$, a contradiction.

\noindent \textbf{Case 2.} $C_4$ lies in inside of $C_2$ (See
Figure\ref{Figure P_4-1}(3), (4)). Then each edge of $E_{4,3}$
crosses the edges of $E_{2}$ at least one time. By Lemma\ref{Lemma
[E'_{i}]}, $\nu_D(\overline{E'_{2}})\geq 2$. It follows $\nu(D)\geq
2+2+2=6$, a contradiction.

\noindent \textbf{Case 3.} $C_4$ lies in outside $C_1$, $C_2$ and
$C_3$ (See Figure\ref{Figure P_4-1}(5)). By Lemma\ref{Lemma
[E'_{i}]}, $\nu_D(\overline{E'_{1}})\geq 3$. Since $\nu(D)\leq 5$,
$C_1$ does not cross itself, any edge of $\bigcup_{j=2,3,4} E_{1,j}$
is clean. It follows edges $v_{1}v_{13}$ and $v_4v_{16}$ are clean.
Now at least one edge of $E_{1,4}$ is crossed, it contradicts any
edge of $\bigcup_{j=2,3,4} E_{1,j}$ is clean.

\begin{figure}[ht]
\center
\includegraphics[scale=1.0]{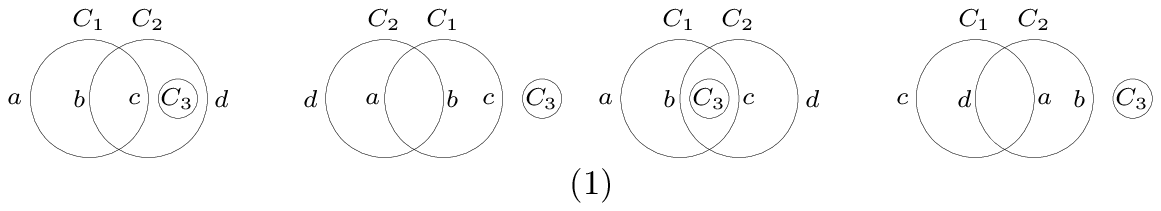}
\includegraphics[scale=1.0]{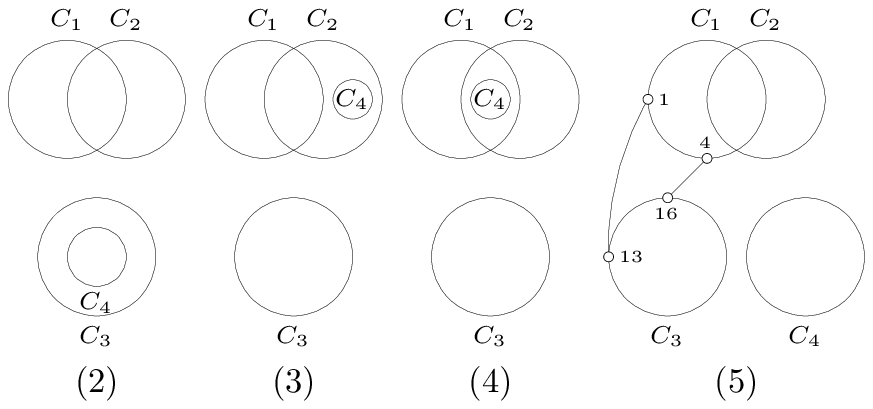}
\caption{Some Drawings of $P_4$, where just one pair of 6-cycles
crosses each other}\label{Figure P_4-1}
\end{figure}

\end{proof}

\begin{lemma}\label{Lemma P_{0}}
 Let $D$ be a drawing of $P_4$, where any pair of $6$-cycles
does not cross each other, then $\nu(D)\geq 6$.
\end{lemma}

\begin{proof}By contradiction. Suppose  $\nu(D)\leq 5$.

\noindent \textbf{Case 1.} $C_2$ lies in inside of $C_1$, $C_3$ and
$C_4$ lie in outside of $C_1$.

\noindent \textbf{Case 1.1.} $C_4$ lies in inside of $C_3$. Then
each edge of $E_{2,4}$ crosses the edges of $E_{1}$ at least one
time, crosses the edges of $E_{3}$ at least one time. Meanwhile,
each edge of $E_{2,3}$ crosses the edges of $E_{1}$ at least one
time, and each edge of $E_{4,1}$ crosses the edges of $E_{3}$ at
least one time. It follows $\nu(D)\geq 4+2+2= 8$ (See
Figure\ref{Figure P_4-0}(1)).

\noindent \textbf{Case 1.2.} $C_4$ lies in outside of $C_3$. Then
each edge of $E_{2,3}$ crosses the edges of $E_{1}$ at least one
time, while each edge of $E_{2,4}$ crosses the edges of $E_{1}$ at
least one time. By Lemma\ref{Lemma [E'_{i}]},
$\nu_D(\overline{E'_{1}})\geq 3$. Then, it follows $\nu(D)\geq
3+4=7$, a contradiction. (See Figure\ref{Figure P_4-0}(2)).

\noindent \textbf{Case 2.} each $C_i$ lies in outside of other three
$C_j$($1\leq j\leq 4, j\neq i$).

\noindent \textbf{Case 2.1.} $C_3$ does not cross itself. Since
$\nu(D)\leq5$ and $\nu_D(\overline{E'_{3}})\geq 3$,
$\nu_D(E'_3)+\nu_D(E'_3,\overline{E'_{3}})\leq 2$. It follows that
$\nu_D(E_{1,3})+\nu_D(E_{3},E_{1,3})+\nu_D(E_{2,3})+\nu_D(E_{3},E_{2,3})+\nu_D(E_{3,4})+\nu_D(E_{3},E_{3,4})\leq2$.
Without loss of generality, we may assume
$\nu_D(E_{1,3})+\nu_D(E_{3},E_{1,3})=0$ (See Figure\ref{Figure
P_4-0}(3)). Then at least one edge of $E_{3,2}$ crosses one edge of
$E_1\cup E_3\cup E_{1,3}$. After that, at least one edge of
$E_{3,4}$ crosses one edge of $E_1\cup E_3\cup E_{1,3}$ and at least
one edge of $E_{3,4}$ crosses one edge of $E_2\cup E_3\cup E_{2,3}$.
It follows $\nu_D(E'_3)+\nu_D(E'_3,\overline{E'_{3}})\geq 3$, which
contradicts $\nu_D(E'_3)+\nu_D(E'_3,\overline{E'_{3}})\leq 2$.

\noindent \textbf{Case 2.2.} $C_3$ crosses itself. Since
$\nu(D)\leq5$ and $\nu_D(\overline{E'_{3}})\geq 3$,
$\nu_D(E'_3)+\nu_D(E'_3,\overline{E'_{3}})\leq 2$. It follows that
$\nu_D(E_{1,3})+\nu_D(E_{3},E_{1,3})+\nu_D(E_{2,3})+\nu_D(E_{3},E_{2,3})+\nu_D(E_{3,4})+\nu_D(E_{3},E_{3,4})\leq1$
since $C_3$ crosses itself. Without loss of generality, we may
assume
$\nu_D(E_{1,3})+\nu_D(E_{3},E_{1,3})=\nu_D(E_{2,3})+\nu_D(E_{3},E_{2,3})=0$.
If $C_1$ does not cross itself, then the two edges of $E_{1,4}$
cross at least three times in total (See Figure\ref{Figure
P_4-0}(4)). If $C_1$ crosses itself, then each edge of $E_{1,4}$
crosses the edges of $E_2\cup E_3\cup E_{2,3}$ at least one time
(See Figure\ref{Figure P_4-0}(5)). By Lemma\ref{Lemma [E'_{i}]},
$\nu_D(\overline{E'_{1}})\geq 3$. It follows $\nu(D)\geq 3+3=6$, a
contradiction.
%
%
%
%

\begin{figure}[ht]
\center
\includegraphics[scale=1.0]{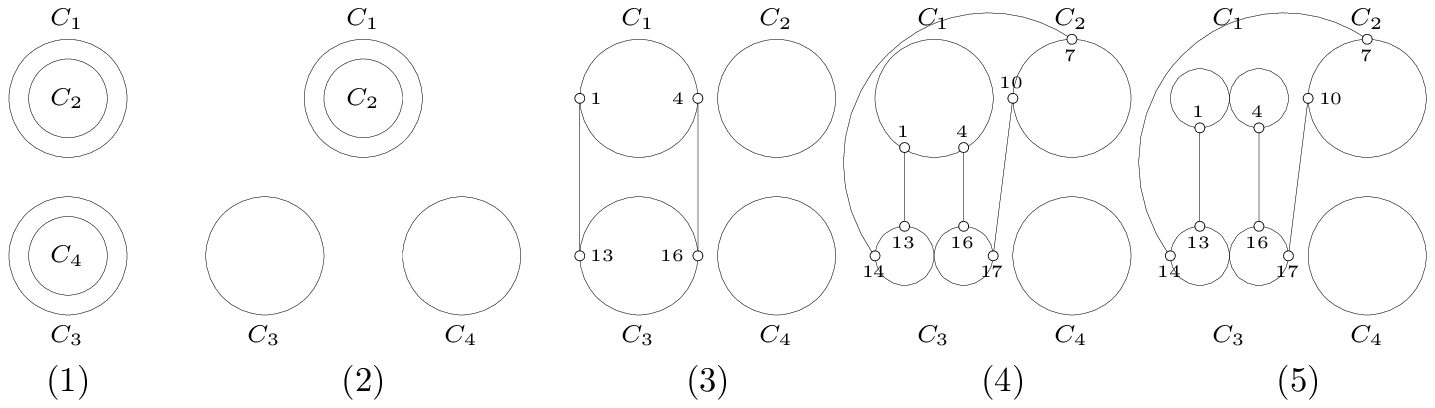}
\caption{Some Drawings of $P_4$, where any pair of $6$-cycles does
not cross each other}\label{Figure P_4-0}
\end{figure}

\end{proof}
By Lemmas \ref{lemma upper bound of P_4}, \ref{Lemma P_{2}} -
\ref{Lemma P_{0}}, we have
\begin{theorem}\label{theorem cr(P4)=6}
$cr(P_4)=6$.
\end{theorem}


\begin{thebibliography}{99}

\bibitem{LHH05}
C.-K. Lin, H.-M. Huang, L.-H. Hsu,
\newblock The super connectivity of the
pancake graphs and the super laceability of the star graphs,
\newblock {\it Theoretical Computer Science} 339 (2005) 257 ¨C 271.

\bibitem{S05}
C.-N. Hung, H.-C. Hsu, K.-Y. Liang, L.-H. Hsu,
\newblock Ring embedding in faulty pancakegraphs,
\newblock {\it Information Processing Letters} 86 (5) 2003 271-275.

\bibitem{DZ12}
Yun-Ping Deng, Xiao-Dong Zhang,
\newblock Automorphism Groups of the Pancake Graphs,
\newblock {\it Information Processing Letters} 112 (7) (2012) 264-266

\bibitem{NB09}
Q. T. Nguyen, S. Bettayeb,
\newblock The upper bound and lower bound of the genus of pancake graphs,
\newblock {\it  Computers and Communications, 2009. ISCC 2009. IEEE Symposium on} (2009): 125-129.

\bibitem{LTHHH09}
C.-K. Lin, J. J.M. Tan, H.-M. Huang, D. F. Hsu, L.-H. Hsu,
\newblock Mutually independent hamiltonian cycles for the pancakegraphs and the star graphs,
\newblock {\it Discrete Math.} 309 (17) (2009): 5474-5483.

\bibitem{C11}
P. E.C. Compeau,
\newblock Girth of pancakegraphs,
\newblock {\it Discrete App. Math.} 159 (15) (2011): 1641-1645.

\bibitem{Bruce95}
A. M. Dean, R. B. Richter,
\newblock The crossing number of $C_4\square C_4$,
\newblock {\it Journal of Graph Theory} 19 (1995): 125-129.

\bibitem{Akers89}
S. B. Akers, B. Krishnamurthy,
\newblock A group-theoretic model for symmetric interconnection networks,
\newblock {\it IEEE Transactions on Computers} 38 (1989): 555-566.

\bibitem{EG73}
P. Erd\H{o}s, R. K. Guy,
\newblock Crossing number problems,
\newblock {\it Amer. Math. Monthly} 80 (1973) 52-58.

\bibitem{Garey83}
M. R. Garey, D. S. Johnson,
\newblock Crossing numbers is NP-complete,
\newblock {\it SIAM J. Alg. Disc. Math.} 4 (1983) 312-316.

\bibitem{Guy60}
R. K. Guy,
\newblock A combinatorial problem,
\newblock {\it Bull. Malayan Math. Soc.} 7 (1960) 68-72.

\bibitem{LinXH09}
X. Lin, Y. Yang, W. Zheng, L. Shi, W. Lu
\newblock The crossing numbers of generalized Petersen graphs with small order,
\newblock {\it Discrete App. Math.} 157 (2009) 1016-1023.

\bibitem{Bruce07}
S. Pan, R. B. Richter,
\newblock The crossing number of $K_{11}$ is 100,
\newblock {\it Journal of Graph Theory} 56 (2007): 128-134.

\bibitem{Tutte70}
W. T. Tutte,
\newblock Toward a theory of crossing numbers,
\newblock {\it J. Combinatorial Theory} 8 (1970) 45-53.

\bibitem{Turan77}
P. Tur\'{a}n,
\newblock A note of welcome,
\newblock {\it J. Graph Theory} 1 (1977) 7-9.

\bibitem{YangYS03}
Y. Yang, Y. Sun, W. Lu,
\newblock Crossing numbers of graphs with at most nine vertices,
\newblock {\it Mini- Microsystems} 24 (6) (2003) 954-958.

\bibitem{YangYS02}
Y. Yang, D. Wang, W. Lu,
\newblock The crossing number of 4-regular Graphs,
\newblock {\it Journal of Software} 13 (12) (2002) 2259-2266.


\bibitem{SV94}
O. S\'{y}kora, I. Vrt'o,
\newblock On VLSI layouts of the star graph and related networks,
\newblock {\it Integration, the VLSI Journal} 17 (1) (1994) 83-93.

\end{thebibliography}
\end{document}